\documentclass{article}
\usepackage{graphicx}
\usepackage{latexsym}
\usepackage{amsmath}
\usepackage{amssymb}
\usepackage{stmaryrd}
\usepackage{xspace}
\usepackage{paralist}
\usepackage{enumitem}
\usepackage{tikz}
\usepackage[noadjust]{cite}
\usepackage{thmtools}
\usepackage{thm-restate}

\newtheorem{theorem}{Theorem}

\newenvironment{proof}{\par\noindent {\em Proof.}\ } {\hfill$\Box$\par\medskip}

\def\ojoin{\setbox0=\hbox{$\Join$}%
  \rule[.07ex]{.3em}{.4pt}\llap{\rule[1.38ex]{.3em}{.4pt}}}
\def\leftouterjoin{\mathbin{\ojoin\mkern-6.5mu\Join}}

\newcommand{\IRI}{\ensuremath{\mathbf{I}}\xspace}
\newcommand{\Var}{\ensuremath{\mathbf{X}}\xspace}

\newcommand{\Graph}{G}

\newcommand{\OPT}{\ensuremath{\mathbin{\mathsf{OPT}}}}

\newcommand{\map}{\mu}

\newcommand{\Dom}[1]{\ensuremath{\mathsf{dom}(#1)}}
\newcommand{\vars}[1]{\ensuremath{\mathsf{vars}(#1)}}

\newcommand{\triple}[3]{\ensuremath{(#1, #2, #3)}}
\newcommand{\mset}{\Omega}
\newcommand{\LOJ}{\mathbin{\leftouterjoin}}
\newcommand{\Sem}[2]{\llbracket #2 \rrbracket_{#1}}

\newcommand{\pspace}{{\rm \textsc{PSpace}}\xspace}

\newcommand{\np}{{\rm \textsc{NP}}\xspace}
\newcommand{\conp}{{\rm \textsc{coNP}}\xspace}
\newcommand{\pitwop}{\ensuremath{\Pi_2^p}}

 \newcommand{\nat}{\mathbb N}
 \newcommand{\Tinst}{\mathbb T}
 \newcommand{\ttypes}{T}
 \newcommand{\tiletype}{t}
 \newcommand{\ttype}{tType}
 \newcommand{\Hrel}{\mathcal H}
 \newcommand{\Vrel}{\mathcal V}
 \newtheorem{fact}{Fact}
 \newcommand{\btype}{\mathit{bType}}
 \newcommand{\ctype}{\mathit{cType}}
 \newcommand{\cellvar}{?\mathit{c}}
 \newcommand{\cell}{\mathit{Cell}}
 \newcommand{\initvar}{?\mathit{r}}
 \newcommand{\itype}{\mathit{hType}}
 \newcommand{\initrow}{\mathit{inInitRow}}
 \newcommand{\hnext}{\mathit{hNext}}
 \newcommand{\vnext}{\mathit{vNext}}
 \newcommand{\basevar}{?\mathit{b}}
 \newcommand{\baseconstA}{\mathit{b_{\sqsubseteq}}}
 \newcommand{\baseconstB}{\mathit{b_{\not\sqsubseteq}}}
 \newcommand{\subsA}{\mathit{Base_{\sqsubseteq}}}
 \newcommand{\subsB}{\mathit{Base_{\not\sqsubseteq}}}
 \newcommand{\sqvarone}{?\mathit{s}_1}
 \newcommand{\sqvartwo}{?\mathit{s}_2}
 \newcommand{\sqvarthree}{?\mathit{s}_3}
 \newcommand{\sqvarfour}{?\mathit{s}_4}
 \newcommand{\Broot}{B_{\text{root}}}
 \newcommand{\BnvH}[1]{B^{#1}_{\text{h-incompat}}}
 \newcommand{\BnvV}[1]{B^{#1}_{\text{v-incompat}}}
 \newcommand{\Btiling}[1]{B^{#1}_{\text{tiling}}}
 \newcommand{\Bbase}{B_{\text{base}}}
 \newcommand{\tilevarone}{?\mathit{tile}_1}
 \newcommand{\tilevartwo}{?\mathit{tile}_2}

\begin{document}

\title{Subsumption of Weakly Well-Designed SPARQL Patterns is Undecidable}
\author{Mark Kaminski \\ \small{Department of Computer Science,} \\ \small{University of Oxford,} \\ \small{Oxford, UK} \\ \small{mark.kaminski@cs.ox.ac.uk} \and Egor V.~Kostylev \\ \small{Department of Computer Science,} \\ \small{University of Oxford,} \\ \small{Oxford, UK} \\ \small{egor.kostylev@cs.ox.ac.uk}}

\maketitle

The Resource Description Framework
(RDF)~\cite{RDFSyntax14,RDFSemantics14}
is the W3C standard for representing linked data on the Web.  
SPARQL~\cite{SPARQL10,SPARQL11} is the default query language for RDF
graphs.

A distinctive feature of SPARQL is
the $\mathsf{OPTIONAL}$ operator (abbreviated as $\OPT$ in this paper), which was introduced to
\textsl{``not reject (solutions) because some part of the query
  pattern does not match''}~\cite{SPARQL10}.
The $\OPT$ operator accounts in a natural way for
the open world assumption and the fundamental incompleteness of the
Web. However, evaluating queries that use \OPT{} is computationally
expensive: the corresponding decision problem is $\pspace$-complete~\cite{PerezEtAl09,SchmidtML10}, even if only projection-free queries (i.e., \emph{patterns}) are considered.

P\'erez et al.~\cite{PerezEtAl09} introduced the \emph{well-designed}
fragment of SPARQL queries by imposing a syntactic restriction on the
use of variables in $\OPT$-expressions. On the one hand, well-designed patterns have
lower complexity of query evaluation---the problem is \conp-complete.  On the other hand, such queries
have a more intuitive behaviour than arbitrary SPARQL queries and enjoy
specific monotonicity properties. However, by far not all SPARQL queries are well-designed~\cite{PicalausaVansummeren11}.
\emph{Weakly well-designed} SPARQL fragment has been recently introduced to overcome this shortcoming: it possesses the same complexity of evaluation, but also includes almost all queries that appear in practice~\cite{DBLP:conf/icdt/KaminskiK16,journals/tocs/KaminskiK18}.

Besides evaluation, every query language has associated static analysis problems, such as \emph{query equivalence} and \emph{containment}.
For SPARQL there is also a specific static analysis problem, namely, \emph{query subsumption}~\cite{LetelierPPS13}. 
It is known that equivalence and containment are both $\np$-complete for well-designed patterns, while subsumption is \pitwop-complete for such queries~\cite{LetelierPPS13,PichlerS14}. 
Moreover, all three problems are undecidable for well-designed queries
with projection~\cite{LetelierPPS13,PichlerS14}.
From the results of Zhang et al.~\cite{ZhangBP16} it follows that all these problems are undecidable for arbitrary patterns.
Finally, equivalence and containment for weakly well-designed patterns
are both $\pitwop$-complete~\cite{DBLP:conf/icdt/KaminskiK16,journals/tocs/KaminskiK18}.
It is also claimed that subsumption is also $\pitwop$-complete for such patterns~\cite{DBLP:conf/icdt/KaminskiK16}. In this paper, however, we show that this problem is much more difficult; in fact, it is undecidable.

\section{SPARQL Patterns}
\label{sec:def}
We adopt the formalisation of SPARQL that mostly follows~\cite{PerezEtAl09}. However, we concentrate on patterns constructed using only basic graph patterns and optional matching.

\smallskip\noindent\textbf{RDF Graphs}~~
An RDF graph is a labelled graph where nodes can also serve as edge
labels. Formally, let $\IRI$ be a set of \emph{IRIs}. Then an
\emph{RDF triple} is a tuple $(s, p, o)$ from
$\IRI\times\IRI\times\IRI$, where $s$ is called \emph{subject}, $p$
\emph{predicate}, and $o$ \emph{object}. An \emph{RDF graph} is a
finite set of RDF triples.

\smallskip\noindent\textbf{SPARQL Syntax}~~
Let $\Var$ be an infinite
set $\{?x, ?y, \ldots\}$ of \emph{variables}, disjoint from $\IRI$.
A \emph{basic (graph) pattern} is a possibly empty set of triples from 
$$
{(\IRI \cup \Var)\times(\IRI\cup \Var)\times(\IRI \cup \Var)}.
$$
An (\emph{optional SPARQL graph) patterns} $P$ are defined by the following grammar, where $B$ ranges over basic patterns:
\begin{align*}
  P & \;\; ::=  \;\; B \mid (P \OPT P).
\end{align*}
We denote $\vars{P}$ the set of all variables that appear in a pattern $P$.

Note that a given pattern can occur more than once within a larger
pattern.  In what follows we will need to distinguish
between a (sub-)pattern $P$ as a possibly repeated building block of
another pattern $P'$ and its \emph{occurrences} in $P'$---that is,
unique subtrees in the parse tree. Then, the \emph{left (right)
  argument} of an occurrence $i$ is the subtree rooted in the left
(right) child of the root of $i$ in the parse tree, and an occurrence
$i$ is \emph{inside} an occurrence $j$ if the root of $i$ is a descendant of the root of $j$.

 A pattern $P$ is \emph{well-designed} (P{\'e}rez et al.~\cite{PerezEtAl09})
  if for every occurrence $i$ of an
  $\OPT$-pattern $P_1 \OPT P_2$ in $P$ the variables from $\vars{P_2}
  \setminus \vars{P_1}$ occur in $P$
  only inside $i$. 

Given a pattern $P$, an
occurrence $i_1$ in $P$ \emph{dominates} an occurrence $i_2$ if
there exists an occurrence $j$ of an $\OPT$-pattern such that $i_1$ is
inside the left argument of $j$ and $i_2$ is inside the right
argument. A pattern $P$ is \emph{weakly well-designed}~(\cite{DBLP:conf/icdt/KaminskiK16,journals/tocs/KaminskiK18}) if, for each occurrence $i$ of an
$\OPT$-subpattern $P_1 \OPT P_2$, the variables in $\vars{P_2} \setminus
\vars{P_1}$ appear outside $i$ only in subpatterns whose occurrences are dominated by $i$. 

\smallskip\noindent\textbf{SPARQL Semantics}~~
The semantics of graph patterns is defined in terms of
\emph{mappings}---that is, partial functions from variables to
IRIs. The \emph{domain} $\Dom{\map}$ of a mapping $\map$ is the set of
variables on which $\map$ is defined.  Two mappings $\map_1$ and
$\map_2$ are \emph{compatible}, written $\map_1 \sim \map_2$, if
$\map_1(?x) = \map_2(?x)$ for all variables $?x\in\Dom{\map_1}\cap\Dom{\map_2}$. 
Mapping $\mu_1$ \emph{is subsumed} by mapping $\mu_2$, written $\mu_1  \sqsubseteq \mu_2$, if $\mu_1 \sim \mu_2$ and $\Dom{\map}_1 \subseteq \Dom{\map_2}$.
If $\map_1\sim\map_2$, then $\map_1\cup\map_2$
constitutes a mapping with domain $\Dom{\map_1}\cup\Dom{\map_2}$ that coincides with $\map_1$ on $\Dom{\map_1}$
and with $\map_2$ on $\Dom{\map_2}$.

Given two sets of mappings $\mset_1$ and $\mset_2$, we define their \emph{left outer join} operation as follows:
\begin{multline*}
\mset_1 \LOJ \mset_2 = \{\map_1 \cup \map_2 \mid \map_1 \in \mset_1, \map_2 \in \mset_2, \text{ and } \map_1 \sim \map_2\} \cup {} \\ \{\map_1 \mid \map_1 \in \mset_1,  \map_1 \not\sim \map_2 \text{ for all }\map_2 \in \mset_2\}.
\end{multline*}
Given a  graph $\Graph$, the \emph{evaluation} $\Sem{G}{P}$ of a pattern $P$ over $G$ is defined as follows:
\begin{enumerate}[noitemsep,topsep=2pt]
\item if $B$ is a basic pattern, then 
$
\Sem{G}{B}=\{\map : \vars{B}\rightarrow\IRI \mid \map(B) \subseteq G\};
$
\item $\Sem{G}{(P_1\OPT P_2)}=\Sem{G}{P_1}\LOJ\Sem{G}{P_2}$.
\end{enumerate}

A pattern $P$ is \emph{contained} in a pattern $P'$ if $\Sem{G}{P} \subseteq \Sem{G}{P'}$ for every graph $G$. Patterns $P$ and $P'$ are \emph{equivalent} if they contain each other. 
Pattern $P$ is \emph{subsumed} by $P'$, written $P \sqsubseteq P'$, if, for every graph~$G$, each $\mu \in \Sem{G}{P}$
has $\mu' \in \Sem{G}{P'}$ such that $\mu \sqsubseteq \mu'$ (Letelier~et~al.~\cite{LetelierPPS13}).

\section{Pattern Subsumption}

 \begin{theorem}
 \label{lem:sub_Und}
  The problem of checking whether $P \sqsubseteq P'$ for weakly well-designed patterns $P$ and $P'$ is undecidable.
 \end{theorem}

 \begin{proof}
 We prove undecidability by a reduction of a variant of the tiling problem, which is known to be undecidable (see e.g., \cite{GK72}). We start by introducing the notation used throughout the proof.

 A \emph{tiling instance} $\Tinst$ consists of a collection $\ttypes = \{\tiletype_1,\ldots ,\tiletype_n\}$
 of {\em tile types} and {\em edge compatibility relations} $\Hrel$ and~$\Vrel$ on $\ttypes$. 
 Intuitively, $\Hrel(\tiletype, \tiletype')$ means that a tile of type $\tiletype'$ can be placed to the right of
 a tile of type $\tiletype$ in a row, while $\Vrel(\tiletype, \tiletype')$ means that $\tiletype'$ can be 
 placed above $\tiletype$ in a column.

 A \emph{tiling} of the positive plane with $\Tinst$ is a function
 $\tau : \nat \times \nat \rightarrow \ttypes$, for the set of natural numbers $\nat$, such that, for all $i,j \in \nat$,
 \begin{itemize}[noitemsep,topsep=2pt]
 \item[--] $\Hrel(\tau(i,j),\tau(i+1,j))$, and
 \item[--] $ \Vrel(\tau(i,j), \tau(i,j+1))$.
 \end{itemize}
 Tiling $\tau$ is \emph{periodic} if there exist positive numbers $p$ and $q$, called \emph{horizontal} and \emph{vertical periods}, respectively, such that $\tau(i,j) = \tau(p+i,j) = \tau(i,q + j)$ for all $i,j \in \nat$. A periodic tiling can be seen as a tiling of a torus, since column $p+1$ and row $q+1$ can be ``glued'' with the left-most column and bottom row, respectively.

 Let $S_{\textrm{tiling}}$ denote the set of all tiling instances that allow for tilings of the positive plane, and $S_{\textrm{period}}$ the set of all tiling instances that allow for periodic tilings.
 To prove undecidability we will use the following fact.

 \begin{fact}[Gurevich and Koryakov~\cite{GK72}]
 \label{f:tile}
 Sets $S_{\textrm{\em tiling}}$ and $S_{\textrm{\em period}}$ are \emph{recursively inseparable}---that is, there is no recursive set $S$ with 
 ${S_{\textrm{\em period}} \subseteq S \subseteq S_{\textrm{\em tiling}}}$.
 \end{fact}

 In what follows we first construct, for each tiling instance $\Tinst$, weakly well-designed patterns $P_\Tinst$ and
 $P'_\Tinst$, and then show that the set
 \begin{equation*}
 \label{eq_phi}
 \Phi = \{(P_\Tinst, P'_\Tinst) \mid P_\Tinst \not\sqsubseteq P'_\Tinst\}
 \end{equation*} 
 contains $\{(P_\Tinst, P'_\Tinst) \mid \Tinst \in S_{\textrm{period}}\}$, 
 and is contained in ${\{(P_\Tinst, P'_\Tinst) \mid \Tinst \in S_{\textrm{tiling}}\}}$. This will imply, by Fact~\ref{f:tile}, that $\Phi$ (and, hence, the complement of $\Phi$) cannot be recursive.

 \medskip

 Let $\Tinst$ be a tiling instance with tile types $\ttypes = \{\tiletype_1,\ldots ,\tiletype_n\}$, and compatibility relations $\Hrel$ and~$\Vrel$.
 Let $P_\Tinst$ be
 $$
 \begin{array}{llllll}
  &  & \{ & \triple{c_{11}}{\itype}{\initrow}, \triple{c_{11}}{\ctype}{\cell},  \\
 & & & \triple{c_{11}}{\hnext}{c_{12}}, \triple{c_{11}}{\vnext}{c_{21}}, \triple{c_{12}}{\vnext}{c_{22}}, \\
 & & & \triple{\basevar}{\btype}{\subsA}~~~ \};
 \end{array}
 $$
 so, $P_\Tinst$ is a basic pattern with 6 triples, only one of which mentions a variable, $\basevar$.
 The other pattern has a more complex structure: let $P'_\Tinst$ be
 \begin{equation}
 \begin{array}{lllllllll}
  \multicolumn{7}{l}{( \cdots (( \cdots (( \cdots ( \Broot \OPT {}} \\
   & \qquad \BnvH{1}) & \OPT & \cdots & \OPT & \BnvH{\ell}) & \OPT \\
   & \qquad ~ \BnvV{1}) & ~ \OPT & ~ \cdots & ~ \OPT & ~ \BnvV{m}) & ~ \OPT \\
   & \qquad ~~ \Btiling{1}) &  ~~ \OPT & ~~ \cdots & ~~ \OPT & ~~ \Btiling{n}) & ~~ \OPT \\
   & \qquad ~~~ \Bbase,
 \end{array}
 \label{Pprime}
 \end{equation}
 where $\ell = |(\ttypes \times \ttypes) \setminus \Hrel|$, $m = |(\ttypes \times \ttypes) \setminus \Vrel|$,
 $$
 \begin{array}{rlll}
 \Broot & = & \{ \, \triple{\initvar}{\itype}{\initrow}, \\
 & & ~~ \triple{\cellvar}{\ctype}{\cell}, \\
 & & ~~ \triple{\sqvarone}{\hnext}{\sqvartwo}, \triple{\sqvarone}{\vnext}{\sqvarthree}, \triple{\sqvartwo}{\vnext}{\sqvarfour} \, \}, \\ \\

 \BnvH{i} & = & \{ \, \triple{\basevar}{\btype}{\subsA}, \\
 & & ~~ \triple{\tilevarone}{\hnext}{\tilevartwo}, \triple{\tilevarone}{\ttype}{\tiletype_1^i},\triple{\tilevartwo}{\ttype}{\tiletype_2^i} \, \}, \\
 & & ~ \qquad \qquad \text{ for each } i = 1, \ldots, \ell, \\
 & & ~ \qquad \qquad \text{ where } (\tiletype_1^i, \tiletype_2^i) \text{ is the } i\text{'th pair in } (\ttypes \times \ttypes) \setminus \Hrel, \\ \\

 \BnvV{j} & = & \{ \, \triple{\basevar}{\btype}{\subsA}, \\
 & & ~~ \triple{\tilevarone}{\vnext}{\tilevartwo}, \triple{\tilevarone}{\ttype}{\tiletype_1^j},\triple{\tilevartwo}{\ttype}{\tiletype_2^j} \, \}, \\ 
 & & ~ \qquad \qquad \text{ for each } j = 1, \ldots, m, \\
 & & ~ \qquad \qquad \text{ where } (\tiletype_1^j, \tiletype_2^j) \text{ is the } j\text{'th pair in } (\ttypes \times \ttypes) \setminus \Vrel, \\ \\

 \Btiling{k} & = & \{ \, \triple{\basevar}{\btype}{\subsB}, \\
 & & ~~  \triple{\initvar}{\ctype}{\cell}, \triple{\initvar}{\hnext}{\initvar'}, \triple{\initvar'}{\itype}{\initrow}, \\
 & & ~~  \triple{\cellvar}{\ttype}{\tiletype_k}, \triple{\cellvar}{\vnext}{\cellvar'}, \triple{\cellvar'}{\ctype}{\cell}, \\
 & & ~~  \triple{\sqvarthree}{\hnext}{\sqvarfour} \, \}, \\ 
 & & ~ \qquad \qquad \text{ for each } k = 1, \ldots, n, \\ \\

 \Bbase & = & \{ \, \triple{\basevar}{\btype}{\subsA} \, \}.
 \end{array}
 $$

 Having the construction complete, next we show that $P_\Tinst \not\sqsubseteq P'_\Tinst$ for any tiling instance $\Tinst$ in $S_{\textrm{period}}$. In particular, on the base of a witnessing periodic tiling we build a graph $\Graph$ and a mapping $\map$ such that $\map \in \Sem{\Graph}{P_\Tinst}$, but there is no $\map' \in \Sem{\Graph}{P'_\Tinst}$ such that $\map \sqsubseteq \map'$. Assume that  $\Tinst$ has tile types $\ttypes = \{\tiletype_1,\ldots ,\tiletype_n\}$, compatibility relations $\Hrel$ and~$\Vrel$, and periodic tiling $\tau$ with the horizontal and vertical periods $p \geq 2$ and $q \geq 2$, respectively.
 Let $\Graph$ consist of the triples
 $$
 \triple{\baseconstA}{\ttype}{\subsA}, \triple{\baseconstB}{\ttype}{\subsB},
 $$
 as well as the triples
 $$
 \begin{array}{ll}
 \triple{c_{1j}}{\itype}{\initrow}, & \text{ for each } j = 1, \ldots, q, \\
 \triple{c_{ij}}{\ctype}{\cell}, & \text{ for each } i = 1, \ldots, p \text{ and } j = 1, \ldots, q, \\
 \triple{c_{ij}}{\ttype}{\tau(i,j)}, & \text{ for each } i = 1, \ldots, p \text{ and } j = 1, \ldots, q, \\
 \triple{c_{ij}}{\hnext}{c_{i(j+1)}}, & \text{ for each } i = 1, \ldots, p \text{ and } j = 1, \ldots, q - 1, \\
 \triple{c_{iq}}{\hnext}{c_{i1}}, & \text{ for each } i = 1, \ldots, p, \\
 \triple{c_{ij}}{\vnext}{c_{(i+1)j}}, & \text{ for each } i = 1, \ldots, p - 1 \text{ and } j = 1, \ldots, q, \\
 \triple{c_{pj}}{\vnext}{c_{1j}}, & \text{ for each } j = 1, \ldots, q.
 \end{array}
 $$
 Let also $\map = \{\basevar \mapsto \baseconstA\}$.

 It is immediate to see that $\map \in \Sem{\Graph}{P_\Tinst}$. Moreover, assuming that $P_\Tinst$ has form \eqref{Pprime}, $\Sem{\Graph}{\Broot}$ consists of $q \cdot (p \cdot q) \cdot (p \cdot q)$ mappings sending $\initvar$ to one of $c_{1j}$, $\cellvar$ to one of $c_{ij}$, $\sqvarone$ also to one of $c_{ij}$, while $\sqvartwo$, $\sqvarthree$ and $\sqvarfour$ to the IRIs accordingly connected to the value of $\cellvar$ (note that the values of $\initvar$, $\cellvar$, and $\sqvarone$ do not depend on each other).

 Since the tiling agrees with $\Hrel$ and $\Vrel$, none of basic patterns $\BnvH{i}$ and $\BnvV{j}$ has a match in $G$, because each of them requires a pair of horizontally or vertically adjacent cells with incompatible tile types. So, none of the mappings in $\Sem{\Graph}{\Broot}$ are extendable to any of $\BnvH{i}$ and $\BnvV{j}$.
 However, each mapping $\map'_\text{root} \in \Sem{\Graph}{\Broot}$ extends to $\Btiling{k}$ such that $t_k = \tau(i, j)$ with $\map'_\text{root}(\cellvar) = c_{ij}$. In particular, this extension $\map'$ sends $\basevar$ to $\baseconstB$, which implies that $\map \not \sqsubseteq \map'$. Therefore, $\Graph$ and $\map$ are a witness for the required $P_\Tinst \not\sqsubseteq P'_\Tinst$.

 \smallskip

 We continue by showing that $P_\Tinst \not\sqsubseteq P'_\Tinst$ implies $\Tinst \in S_{\textrm{tiling}}$ for any tiling instance~$\Tinst$. In particular, on the base of a graph $\Graph$ and mapping $\map$ witnessing $P_\Tinst \not\sqsubseteq P'_\Tinst$ we construct a tiling $\tau$ of the positive plane with $\Tinst$. Assume that  $\Tinst$ has tile types $\ttypes = \{\tiletype_1,\ldots ,\tiletype_n\}$ as well as compatibility relations $\Hrel$ and~$\Vrel$. Since $\map \in \Sem{\Graph}{P_\Tinst}$, graph $\Graph$ contains triples
 $$
 \begin{array}{llllll}
 \triple{c_{11}}{\itype}{\initrow},\triple{c_{11}}{\ctype}{\cell},  \\
 \triple{c_{11}}{\hnext}{c_{12}}, \triple{c_{11}}{\vnext}{c_{21}}, \triple{c_{12}}{\vnext}{c_{22}}, \\
 \triple{\baseconstA}{\btype}{\subsA}
 \end{array}
 $$
 for the IRI $\baseconstA$ such that $\map = \{\basevar \mapsto \baseconstA\}$. Therefore, assuming that $P'_\Tinst$ has form \eqref{Pprime}, $\Sem{\Graph}{\Broot}$ contains a mapping $\map'_\text{root}$ sending $\initvar$ to $c_{11}$.
 Mapping $\map'_\text{root}$ is extendable to $\Btiling{k}$ for some $k$; indeed, if it is not the case, then $\Sem{\Graph}{P'_{\Tinst}}$ contains an extension $\map'$ of $\map'_\text{root}$ sending $\basevar$ to $\baseconstA$, because all $\BnvH{i}$, $\BnvV{j}$, and $\Bbase$ contain $\triple{\basevar}{\ttype}{\subsA}$, while $\Bbase$ matches $\Graph$, which implies $\map \sqsubseteq \map'$ contradicting the fact that $\Graph$ and $\map$ are a witness for non-subsumption. Therefore, triples $\triple{\initvar}{\ctype}{\cell}, \triple{\initvar}{\hnext}{\initvar'}, \triple{\initvar'}{\itype}{\initrow}$ are matched in $\Graph$ extending $\map'_\text{root}$, that is, $\Graph$ contains triples
 $$
 \triple{c_{11}}{\ctype}{\cell}, \triple{c_{11}}{\hnext}{c'_{12}}, \triple{c'_{12}}{\itype}{\initrow}
 $$
 for some IRI $c'_{12}$. Just for uniformity, assume that $c'_{12} = c_{12}$. 
 Therefore, $\Sem{\Graph}{\Broot}$ contains a mapping $\map''_\text{root}$ sending $\initvar$ to $c_{12}$ (and all other variables same as $\map'_\text{root}$). Reasoning in the same way as for $\map'_\text{root}$, we obtain that $\Graph$ has triples
 $$
 \triple{c_{12}}{\ctype}{\cell}, \triple{c_{12}}{\hnext}{c_{13}}, \triple{c_{13}}{\itype}{\initrow}
 $$
 for some IRI $c_{13}$. Continuing like this, we conclude that $\Graph$ contains
 $$
 \triple{c_{1j}}{\ctype}{\cell}, \triple{c_{1(j+1)}}{\hnext}{c_{1(j+1)}}
 $$
 for all $j \geq 1$ (note that many of these $c_{1j}$ coincide, because $\Graph$ is finite).

 For each $j \geq 1$, $\Sem{\Graph}{\Broot}$ contains a mapping sending $\cellvar$ to $c_{1j}$. As before, this mapping is extendable in $\Graph$ to $\Btiling{k}$ for some $k$. In particular, it is extendable to the triples $\triple{\cellvar}{\ttype}{\tiletype_k}$, $\triple{\cellvar}{\vnext}{\cellvar'}$, and $\triple{\cellvar'}{\ctype}{\cell}$---that is, $\Graph$ contains triples
 $$
 \triple{c_{1j}}{\ttype}{\tiletype_k}, \triple{c_{1j}}{\vnext}{c_{2j}}, \triple{c_{2j}}{\ctype}{\cell}
 $$
 for some IRI $c_{2j}$ (again, if $j$ is 1 or 2, then we assume that $c_{2j}$ is the same as in $P_\Tinst$ for uniformity). Similarly as before, $\Sem{\Graph}{\Broot}$ contains a mapping sending $\cellvar$ to $c_{2j}$, from which we have that $\Graph$ has triples 
 $$
 \triple{c_{2j}}{\ttype}{\tiletype_k}, \triple{c_{2j}}{\vnext}{c_{3j}}, \triple{c_{3j}}{\ctype}{\cell}
 $$
 for some $c_{3j}$ and $k$. Repeating this process, we conclude that $\Graph$ contains, for any $i \geq 1$ and $j \geq 1$,
 $$
 \triple{c_{ij}}{\ttype}{\tiletype_{ij}}, \triple{c_{ij}}{\vnext}{c_{(i+1)j}}, \triple{c_{(i+1)j}}{\ctype}{\cell}
 $$
 for some $c_{ij}$ and $\tiletype_{ij}$. Set $\tau(i, j) = t_{ij}$ for each $i$ and $j$. 

 We need to show that $\tau$ is indeed a tiling with $\Tinst$. To this end, we first note that $\Graph$ contains the triple $\triple{c_{ij}}{\hnext}{c_{i(j+1)}}$ for all $i$ and $j$: we already showed this fact for $i = 1$, and for all other $i$ it can be proved very similarly to the reasoning above, based on the fact that $\Sem{\Graph}{\Broot}$ contains a mapping sending $\sqvarone$, $\sqvartwo$, $\sqvarthree$, and $\sqvarfour$ to $c_{(i-1)j}$, $c_{(i-1)(j+1)}$, $c_{ij}$, and $c_{i(j+1)}$, respectively.
 Now, to see that $\tau$ is a tiling with $\Tinst$ we just note that if there exist horizontally or vertically adjacent tiles that do not agree with $\Hrel$ or $\Vrel$, then there exists $i$ or $j$ such that $\BnvH{i}$ or $\BnvV{j}$ is matched in $\Graph$; since this basic patterns does not have any variables in common with $\Broot$, any mapping in $\Sem{\Graph}{\Broot}$ is then extendable to 
 this BGP and hence $\Sem{\Graph}{P'_{\Tinst}}$ contains a mapping sending $\basevar$ to $\baseconstA$, contradicting the fact that graph $\Graph$ and mapping $\map$ are a witness for non-subsumption. 
 \end{proof}

\bibliographystyle{plain}
\bibliography{bibext}

\end{document}